\documentclass{article}
\usepackage{geometry}
\usepackage[utf8]{inputenc}
\usepackage{comment}
\usepackage{type1cm} 
\usepackage{graphicx} 
\usepackage{xspace} 
\usepackage{balance} 
\usepackage{booktabs} 
\usepackage{multirow} 
\usepackage[font={bf}, tableposition=top]{caption} 
\usepackage{subcaption} 
\usepackage{bold-extra} 
\usepackage{siunitx} 
\usepackage[vlined,linesnumbered,ruled,noend]{algorithm2e} 
\usepackage{microtype} 
\usepackage{xfrac} 
\usepackage{mathtools} 
\usepackage{amssymb} 

\usepackage{url} 
\usepackage[bookmarks, pdftex, colorlinks=true, pagebackref=true, backref=page]{hyperref} 
\usepackage[capitalise]{cleveref} 
\usepackage[numbers]{natbib} 
\usepackage{hyphenat} 
\usepackage[hyperpageref]{backref} 
\PassOptionsToPackage{hyphens}{url} 
\usepackage{bbold}

\usepackage{tikz}
\usepackage{tikz-qtree}
\usetikzlibrary{decorations.pathreplacing}
\usetikzlibrary{shadows}
\usepackage[framemethod=TikZ]{mdframed}
\usetikzlibrary{bayesnet}
\usepackage{wrapfig}
\usepackage{graphicx} 
\usepackage{float}
\usepackage{bbm}
\usetikzlibrary{trees} 
\usepackage{placeins}
\usepackage{blindtext}
\usepackage[export]{adjustbox}
\usepackage[normalem]{ulem}

\graphicspath{{figs}}

\newcommand{\spara}[1]{\smallskip\noindent\textbf{#1}}

\PassOptionsToPackage{dvipsnames}{xcolor}
\hypersetup{
    colorlinks=true,
    linkcolor={red!50!black},
    filecolor={green!50!black},
    citecolor={green!50!black}, 
    urlcolor={blue!80!black},
}
\renewcommand*\backref[1]{\ifx#1\relax \else (Cited on #1) \fi}

\newcommand{\abs}[1]{\ensuremath{\lvert #1 \rvert}\xspace}
\newcommand{\bc}{\ensuremath{\varepsilon}\xspace}

\newcommand{\est}[1]{\ensuremath{\hat{#1}}\xspace}
\newcommand{\init}{\ensuremath{\statevector^0}\xspace}

\newcommand{\order}{\ensuremath{\mathcal{T}}\xspace}
\newcommand{\opspace}{\ensuremath{\mathcal{X}}\xspace}
\newcommand{\rate}{\ensuremath{\mu}\xspace}
\newcommand{\state}{\ensuremath{x}\xspace}
\newcommand{\staterv}{\ensuremath{X}\xspace}
\newcommand{\statevector}{\ensuremath{\mathbf{x}}\xspace}
\newcommand{\statevectorrv}{\ensuremath{\mathbf{X}}\xspace}
\newcommand{\tru}[1]{\ensuremath{{#1}^*}\xspace}
\newcommand{\seedsrasch}{{\ensuremath{4000}}\xspace}

\DeclareMathOperator*{\argmax}{arg\,max}

\usepackage{authblk}
\usepackage{lineno}

\usepackage{amsthm}
\usepackage{thmtools}
\newtheorem{theorem}{Theorem}[section]

\newtheorem{lemma}[theorem]{Lemma}

\title{Bias and Identifiability in the Bounded Confidence Model}
\author[1]{Claudio Borile}
\author[1,2]{Jacopo Lenti}
\author[3]{Valentina Ghidini}
\author[1]{Corrado Monti}
\author[1]{Gianmarco~De~Francisci~Morales}
\affil[1]{CENTAI, Turin, Italy}
\affil[2]{Sapienza University, Rome, Italy}
\affil[3]{Euler Institute, Universit\`a della Svizzera Italiana, Lugano, Switzerland}
\date{}

\begin{document}
\maketitle

\begin{abstract}
Opinion dynamics models such as the bounded confidence models (BCMs) describe how a population can reach consensus, fragmentation, or polarization, depending on a few parameters.
Connecting such models to real-world data could help understanding such phenomena, testing model assumptions.
To this end, estimation of model parameters is a key aspect, and maximum likelihood estimation provides a principled way to tackle it.
Here, our goal is to outline the properties of statistical estimators of the two key BCM parameters: the confidence bound and the convergence rate.
We find that their maximum likelihood estimators present different characteristics:
the one for the confidence bound presents a small-sample bias but is consistent, while the estimator of the convergence rate shows a persistent bias.
Moreover, the joint parameter estimation is affected by identifiability issues for specific regions of the parameter space, as several local maxima are present in the likelihood function.
Our results show how the analysis of the likelihood function is a fruitful approach for better understanding the pitfalls and possibilities of estimating the parameters of opinion dynamics models, and more in general, agent-based models, and for offering formal guarantees for their calibration.

\end{abstract}

\section{Introduction}
\label{sec:intro}
\begin{wrapfigure}{r}{0.33\textwidth}
     \centering
         \includegraphics[trim={1cm 0.7cm 0.5cm 0.7cm},clip,width=0.33\textwidth]{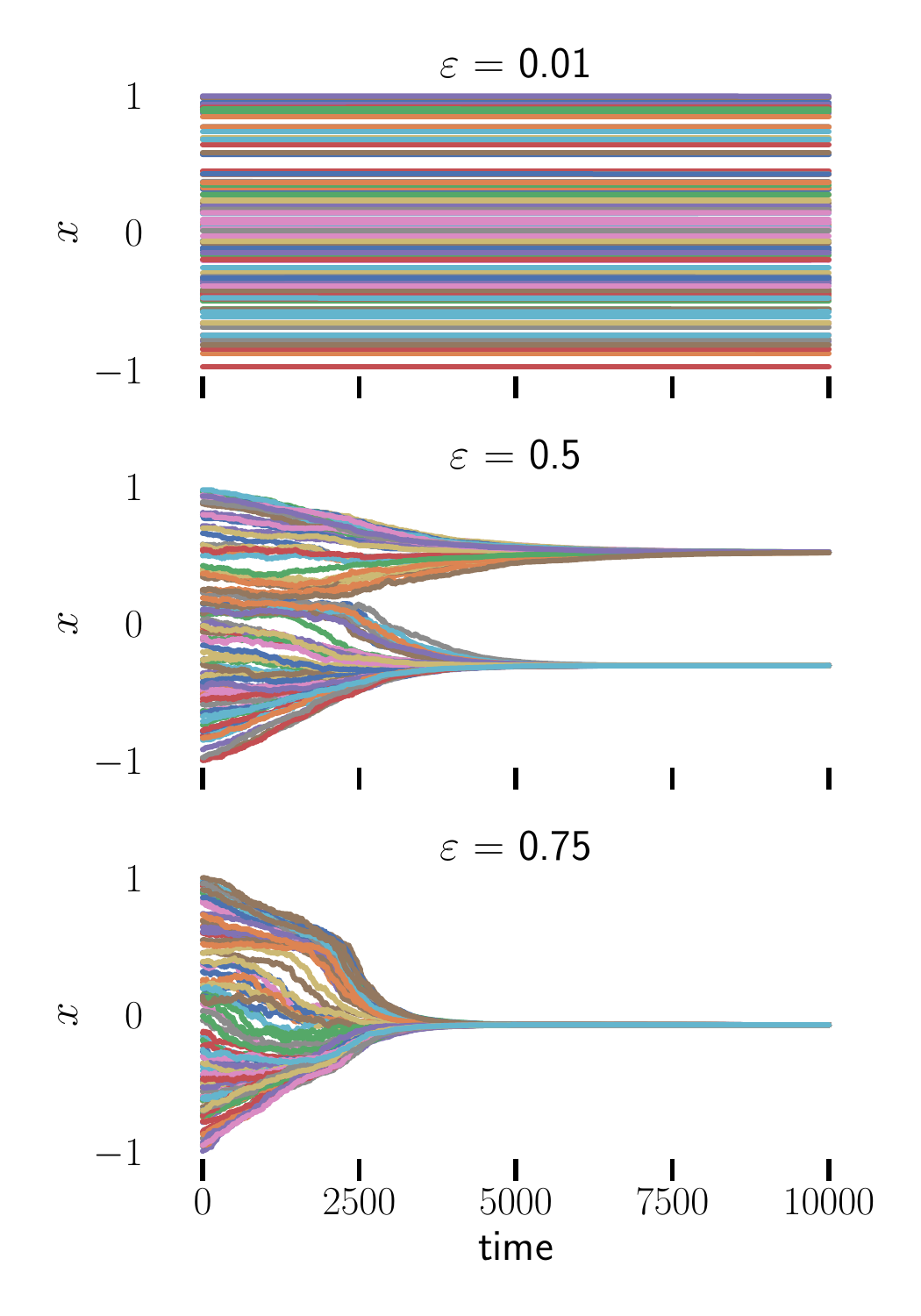}
        \caption{Examples of the evolution of the opinions in the Bounded Confidence Model, varying the \bc parameter.
        }
        \label{fig:PB_example}
\end{wrapfigure}

Opinion dynamics models (ODMs)  aim to uncover the minimal assumptions about individual behavior that yield emergent collective patterns such as consensus or polarization~\cite{hegselmann2002opinion,lorenz2007continuous}.
They are a prominent class of agent-based models (ABMs).
ABMs describe discrete-time dynamical systems via interacting individual agents, with the goal of understanding in this way their emergent macroscopic behavior.
They provide more intuitive specifications over aggregate formalisms such as differential equations, especially for individual-level phenomena, exploring implications of behavioral rules through Monte Carlo simulations.
While these models generate complex dynamics challenging to study analytically, their application is shifting~\citep{carpentras2022mapping,carpentras2023we}: researchers increasingly recognize their potential for modeling and forecasting real-world phenomena~\cite{lavin2021simulation,pangallo2024unequal}.
This shift increases the need for robust parameter estimation from data.

This study confronts the critical challenge of parameter identifiability in Deffuant's Bounded Confidence Model (BCM)~\cite{deffuant2000mixing}, one of the most extensively-studied agent-based models of opinion dynamics~\cite{gomez2012bounded,bernardo2024bounded}.
Traditionally, estimating parameters in ABMs relies on simulation-based approaches~\cite{cranmer2020frontier}.
This often involves an intricate, manual \emph{calibration} process where model outputs are compared against empirical summary statistics~\cite{windrum2007empirical,platt2020comparison}.
This process is often hampered by subjectivity and the inherent technical challenges of interpreting high-dimensional ABM output data~\cite{lee2015complexities}.
Compounding this issue, rigorous studies that verify the reliability of these estimations remain notably scarce.
Indeed, this ``dearth of empirical works''~\cite{flache2017models,castellano2009statistical} highlights the need for more principled model inference methodologies.

A critical aspect of any reliable parameter estimation in ABMs is \emph{identifiability}: the ability to accurately recover true parameter values given sufficient data.
Essentially, for any estimation to be meaningful, distinct parameter values should produce distinguishable data distributions---a vital property for valid calibration and reliable modeling outcomes.
Identifiability is typically categorized into \emph{structural} (theoretical recoverability with ideal data) and \emph{practical} (recoverability considering finite, noisy data)~\cite{wieland2021structural,lee2015complexities,lenti2024likelihoodbased}.
Consequently, a model can be structurally identifiable yet suffer from practical identifiability issues~\citep{kiss2022parameter,gallo2022lack}.
This concern is echoed in studies of epidemic models~\cite{robin2024system}, where practical identifiability can be elusive due to noise and model assumptions, even for mission-critical COVID-19 models~\cite{burger2023computational}.

To address this aspect, we advocate for likelihood-based methods.
While universally important, the principled statistical foundations of likelihood-based inference benefit significantly the formal analysis of identifiability.
This framework offers a level of rigor and analytical tools for studying identifiability that are often difficult to achieve within purely simulation-based techniques.
This approach stems from a recent trend of casting ABMs into probabilistic generative models with well-defined likelihood functions~\cite{monti2020learning,monti2023learning,lenti2024likelihoodbased,lenti2024variational}, thereby granting access to the robust and principled toolkit of statistical inference.

In this paper, we specifically investigate the extent to which the BCM's two fundamental parameters---the confidence bound \bc and the convergence rate \rate~\cite{deffuant2000mixing}---can be accurately estimated from synthetic, model-generated simulation traces.
\Cref{fig:PB_example} illustrates how varying \bc alone can produce diverse macro-level phenomena, underscoring the parameter's importance.
To rigorously tackle identifiability for the BCM, we adopt a maximum likelihood estimation (MLE) framework.
A key methodological step is our development of a probabilistic variant of the BCM (sBCM), a probabilistic interpretation that generalizes the original formulation (see \Cref{sec:models}) and is essential for likelihood-based analysis.
We then investigate the properties of the maximum likelihood estimators, assuming micro-level agent interactions are observable while the macro-level model parameters are latent.

Our analysis addresses both theoretical and practical identifiability.
Theoretically, we examine observational equivalence and estimator consistency.\footnote{The ABM community refers to the problem that different models or their parameterizations can lead to the same output as ``equifinality"~\cite{an2023modeling}.
Some studies have tackled the problem of identifiability with simulation-based approaches, focusing on the global behavior of the system~\cite{bai2022efficient, carrella2021no}.
However, as the ABM community usually focuses on the macroscopic properties arising from the system, equifinality typically maps the connection between parameters and the summary statistics of interest emerging during the simulation.
It is straightforward that lack of microscopic identifiability (the topic of this paper) implies macroscopic equifinality.}
Notably, for the confidence bound \bc, we prove an analytical upper bound to the bias of its MLE and demonstrate its asymptotic unbiasedness (\Cref{theorem:asymptotically_unbiased}).
Practically, we evaluate estimator variance to assess reliability.
Our results demonstrate that the MLE for \bc is asymptotically unbiased, while the estimator for the convergence rate \rate exhibits a persistent bias, a consequence of violating standard MLE regularity conditions as the parameter \rate influences the support of the data-generating distribution (\Cref{sec:anl-est-rate}).
Furthermore, joint parameter estimation can be affected by practical identifiability issues for specific regions of the parameter space.
Ultimately, this work illustrates how maximum-likelihood estimation theory offers a general methodology for assessing and addressing identifiability issues in ABM calibration, highlighting potential pitfalls in the inference of even simple models.

\section{Models}
\label{sec:models}

\spara{Bounded Confidence Models.}
The class of BCMs considers a stochastic process defined on the opinion space of $N$ agents~$\opspace=[-1, 1]^N$.
Given the length of the data trace $T$ (i.e., the number of steps in the stochastic process), the sequence of states yielded by the BCM can be modeled by a sequence of random variables $\statevectorrv^0, \ldots, \statevectorrv^T$ defined on the probability space $\left( \opspace , \mathcal{B}(\opspace) \right)$, where $\mathcal{B}$ is the standard Borel $\sigma$-algebra.
The realizations of such random variables are the states of the process, denoted by $\init, \ldots, \statevector^T$, where $\statevector^t = (\state_1^t, \ldots, \state_N^t) \in \opspace$.

In Deffuant's formulation, at each time step $t \in [T]$, two agents $1 \leq i,j \leq N$ are chosen uniformly at random to interact.
The internal state of agent $i$ is then updated according to the rule

\begin{equation}
\staterv_i^{t+1} \mid \big(\staterv_i^{t} = \state_i^t, \staterv_j^{t} = \state_j^t \big) = \begin{cases}
	\state_i^t + \rate(\state_j^t - \state_i^t) & \text{ if } \abs{\state_j^t - \state_i^t} \leq \bc\\
	\state_i^t & \text{ otherwise; }
\end{cases}
\label{eq:original_bcm}
\end{equation}
and an equivalent update happens for agent $j$.
The parameter $\bc \in [0, 2]$ is called the \emph{confidence bound}, and it represents the maximum distance of opinions such that the two agents are able to interact and change their opinions.
The parameter $\rate \in \left[0, \sfrac{1}{2}\right]$, instead, is the \emph{convergence rate}, which encodes the mobility of opinions.
The process has infinitely many absorbing states~\cite{dubovskaya2023} depending on the value of \bc, that is, all those configurations where the absorbing state presents $C$ clusters of agents with the same opinions such that the inter-distance between any two opinions is larger than \bc.

\spara{Stochastic Bounded Confidence Model (sBCM).}
Since we are interested in the study of statistical estimators, we adopt a probabilistic variant of the Deffuant model~\citep{deffuant2000mixing}.
The stochastic relaxation we study was introduced by \citet{monti2020learning} to fit real data, and it reduces to the original BCM in the limit as $\rho \to \infty$.
Such a generalization allows us to study the inference of the parameters of the model as statistical estimators, while at the same time preserving the original model as an asymptotic limit.
It allows a probabilistic interpretation of the process such that any given data trace has a non-zero probability, however small, of being generated by a given parameterization of the model.
In addition, it enables differentiability of the likelihood function for easier computation of the MLE.
In practice, the discontinuous step-function of the piecewise definition is replaced by a logistic function $\sigma_{\rho}(z) = \sfrac{1}{(1+e^{-\rho z})}$, with the parameter $\rho$ controlling the steepness of the sigmoid.
This way the deterministic update rule becomes a stochastic trial, where the original model is recovered as the sigmoid approaches the step function (for $\rho \to \infty$).
In this sense, it falls into the category of BCM with ``smooth influence''~\citep{deffuant2002can,deffuant2006comparing,steiglechner2024noise,edmonds2005assessing}.

Therefore, we consider the following stochastic process.
At each time step, we draw a pair of agents~$(i, j)$ uniformly from the~$\binom{V}{2}$ possible pairs, thus determining the sequence of agent draws $\order=((i^1, j^1, 1), \ldots, (i^T, j^T, T))$.
Let $E$ represent the set of \emph{successful} interactions.
The triplet $(i, j, t) \in \order$ belongs to $E$ with probability
\begin{equation}
P\left( (i, j, t) \in E \mid \staterv_i^{t} = \state_i^t, \staterv_j^{t} = \state_j^t  \right) = \sigma_{\rho} \left( \bc - \abs{\state_i^t-\state_j^t} \right).
\label{eq:pijt}
\end{equation}
Then, the opinions of both agents are updated as in the original model only if the interaction was successful
 \begin{equation}
 \staterv_i^{t+1} \mid \big(\staterv_i^{t} = \state_i^t, \staterv_j^{t} = \state_j^t \big) = \begin{cases}
	\state_i^t + \rate(\state_j^t - \state_i^t) & \text{ if } (i, j, t) \in E\\
	\state_i^t & \text{ otherwise. }
\end{cases}
\label{eq:modified_bcm}
\end{equation}

\section{Analysis} \label{sec:analysis}

Our goal is to study how the parameters $\Theta=(\bc, \rate)$ of the sBCM can be estimated from observed data with a likelihood-based approach.
We restrict our focus to estimating \bc and \rate, as $\rho$ is a nuisance parameter introduced in the stochastic translation, and is of little modeling interest.

To this end, we develop the Probabilistic Generative Model (PGM) associated with the ABM.
This way, we can obtain and analyze the likelihood function of the latent variables, similarly to \citet{lenti2024likelihoodbased}.
Let us define the latent variables, the observed variables, and the parameters of the PGM.
We first estimate \bc and \rate separately, by considering the other one known; then we consider the scenario when both parameters are latent.
In the sBCM process, the sequence of agents draws \order and the set of successful interactions $E \subseteq \order$ are stochastic, and we assume them observable.
We also consider the initial opinion state \init to be known; the case when \init is not known is far more complex~\citep{monti2020learning} and is deferred to future work.
In summary, the parameters \bc and \rate are latent (one at a time, and then together), while the initial state \init, the interactions \order, and their outcomes $E$ are observable.
Let us then consider the sequence of the realizations of the opinion states $\statevector^{<t} := \init, \ldots, \statevector^{t-1}$.
Each element~$\statevector^t$ can be computed deterministically given~$\statevector^{t-1}, \rate$, and $E^t$ according to \Cref{eq:modified_bcm}.
Therefore,~$\statevector^{<t}$ is a deterministic function of the parameters $\Theta$, the initial conditions~\init, and the sequence of stochastic variables~$E^{<t} \coloneqq 
\{(i, j, t') \in E \mid t' < t\}$.
\Cref{fig:graphical-model} is the graphical representation of the PGM.
This diagram already hints that the estimation process of \bc and \rate is necessarily very different: while \bc has a direct effect on the observed variables $E$, the parameter \rate only connects to the unobserved variables $\{ \statevector^t \}_{t > 0}$ that encode the current state of the system.


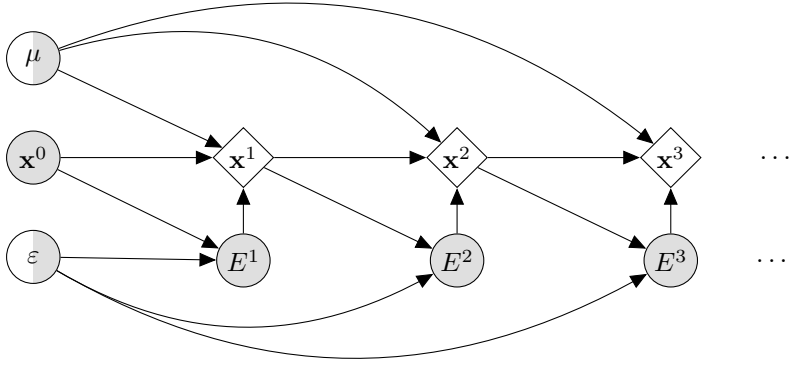
\begin{figure}
\begin{center}
\begin{tikzpicture}[halfcircle/.style={circle, draw=#1, path picture={\fill[#1] (path picture bounding box.north east) rectangle (path picture bounding box.south);}}, halfcircle/.default=gray!25]
\tikzstyle{halflatent} = [halfcircle,fill=white,draw=black,inner sep=1pt,minimum size=20pt, font=\fontsize{10}{10}\selectfont, node distance=1]

\node[obs]                          (x0)     {$\statevector^0$};

\node[halflatent, above=0.6cm of x0]    (mu)     {$\rate$};
\node[halflatent, below=0.6cm of x0]    (eps)    {$\bc$};

\node[det, right=2cm of x0]          (x1)     {$\statevector^1$};
\node[obs, below=0.6cm of x1]        (e1)     {$E^1$};
\edge {x0} {e1}
\edge {x0} {x1}
\edge {e1} {x1}
\edge {eps}{e1}
\edge {mu} {x1}

\node[det, right=2cm of x1]          (x2)     {$\statevector^2$};
\node[obs, below=0.6cm of x2]        (e2)     {$E^2$};
\edge {x1} {e2}
\edge {x1} {x2}
\edge {e2} {x2}
\path[->] (mu) edge[bend left] node [left] {} (x2);
\path[->] (eps) edge[bend right] node [left] {} (e2);

\node[det, right=2cm of x2]          (x3)     {$\statevector^3$};
\node[obs, below=0.6cm of x3]        (e3)     {$E^3$};
\edge {x2} {e3}
\edge {x2} {x3}
\edge {e3} {x3}
\path[->] (mu) edge[bend left] node [left] {} (x3);
\path[->] (eps) edge[bend right] node [left] {} (e3);

\node[right=2cm of x3] (endx) {};
\path (x3) -- node[auto=false]{\ldots} (endx);
\node[right=2cm of e3] (ende) {};
\path (e3) -- node[auto=false]{\ldots} (ende);



\end{tikzpicture}
\end{center}
\vspace{-0.5\baselineskip}
\caption{Graphical model diagram of the stochastic BCM for $T=3$ time-steps. Diamonds indicate deterministic variables, white circles indicate latent stochastic variables, and grey circles indicate observed stochastic variables.
Half gray circles (\bc, \rate) indicates variables that are either latent or observable according to the experimental scenario considered.
}
\label{fig:graphical-model}
\vspace{-\baselineskip}
\end{figure}

\label{sec:role-bc-rate}

\spara{Likelihood function.}
We consider the likelihood function as the probability of observing a certain interaction graph $E\subseteq \order$, that is the set of all agent draws that were successful, given a set of parameters.
Since all interactions are independent, from \Cref{eq:pijt} it follows that such a function is

\begin{equation}
	\mathcal{L}(\Theta \mid \init, E) = %
	P(E \mid \init, \Theta) = %
	\prod_{(i, j, t) \in E}\sigma_{\rho}(\bc-\lvert \state_i^t - \state_j^t \rvert) %
	\prod_{(i, j, t) \in E^c}[1-\sigma_{\rho}(\bc - \lvert \state_i^t - \state_j^t \rvert)] ,
	\label{eq:p_E}
\end{equation}
with $E^c = \order \setminus E$ such that $E\cup E^c=\order$ and  $E\cap E^c=\varnothing$.

\spara{Maximum Likelihood Estimation.}
Given this view of the process, we ask whether the MLEs for the parameters of the sBCM $\Theta=(\bc, \rate)$ are biased. 
Recall that a MLE is a function of the observed data that is used to infer the parameters of its underlying probability distribution, by maximizing its likelihood function to make the observed data most probable.

In the following, we analyze the MLE of \bc and of \rate.
For \bc, we write the complete data likelihood and analyze its derivative with respect to \bc, assuming the other parameter \rate is known.
In this case, the realizations $\statevector^t$ are known and the stochastic events (i.e., $E^t$) are independent of each other given the sequence $\init, \ldots, \statevector^t$.
While the estimator has no closed form, this framework allows us to derive an analytical upper bound to the bias of the estimator.
Such a bias is in practice small as it vanishes at a rate of $T^{-1}$.

On the contrary, the case for \rate is more complex, as we are unable to rely on the standard results for MLEs.
In this case, $\statevector^t$ is a function of \rate, and therefore the stochastic events $E^t$ are no longer independent.
In addition, the likelihood of \rate involves the recursive computation of the sequence $\init, \ldots, \statevector^t$ which is hard to treat analytically.
Finally, we show that both the likelihood of \rate and its support depend on the value of \rate, thus hindering the application of standard results from the statistical theory of MLEs.

Our experiments in \Cref{sec:exp-est-rate} show that the MLE of \rate is indeed biased upwards.
As discussed in \Cref{sec:discussion}, we speculate that the reason for this bias is that only successful interactions (i.e., when $(i, j, t) \in E$) give information about \rate, since \rate affects the evolution of $\statevector^t$ only in this case.
At the same time, successful interactions suggest that the opinions~$\state_1^t, \ldots, \state_N^t$ are converging over time, which is more likely when \rate is higher.

\subsection{Estimation of \texorpdfstring{\bc}{Confidence bound}}
\label{sec:anl-est-bc}

In this section  we derive the MLE of \bc, which maximizes the probability of observing $E$ expressed in \Cref{eq:p_E}. We prove that the estimator is asymptotically unbiased and that the absolute value of the bias is bounded from above by $\mathcal{O}(\frac{1}{\rho T})$. Finally, we provide conditions for the existence of the estimator and its relation with Deffuant's formulation of the BCM.

Let \rate and \init be observed.
Since the opinion update is a deterministic function of the interactions, the previous opinions, and \rate, we can compute $\state_i^t$ for any $i$ and $t$.

\begin{theorem}
    Let $\est{\bc}$ be the MLE of $\bc$. Then, 
    \begin{equation}
    \lvert Bias(\est{\bc})\rvert < \frac{1}{8 \rho T}.  
    \end{equation}
    Thus, $\est{\bc}$ is asymptotically unbiased.
\label{theorem:asymptotically_unbiased}
\end{theorem}

To prove \Cref{theorem:asymptotically_unbiased}, 
we first prove two lemmas, to implicitly derive the MLE and to connect it with Rasch models~\citep{fischer200616}, and then we will give the proof of the theorem.
For the sake of clarity, we define $s_{ij}^t$ as the outcome of interaction between $i$ and $j$ at time $t$.
If $(i,j,t)\in E$, $s_{ij}^t = 1$, otherwise $s_{ij}^t = 0$.

\begin{lemma}
Let \est\bc be the MLE of \bc, and $m = \sum\limits_{(i,j,t) \in \order} s_{ij}^t$ the number of positive interactions, then $\est\bc$ satisfies
\begin{equation}
    \sum\limits_{(i,j,t) \in \order} \sigma_{\rho}(\est\bc - \lvert \state^t_i - \state^t_j \rvert) = m.
\label{eq:mle_bc}
\end{equation}
\label{lemma:sum_prob}
\end{lemma}

\begin{proof}
To prove \Cref{lemma:sum_prob}, we maximize \Cref{eq:p_E} with respect to \bc.
Specifically, we set the derivative of the log-likelihood equal to 0, and we verify that the solution corresponds to a maximum.
For brevity, we define $\kappa_{i, j, t}(\bc) \triangleq \sigma_{\rho}(\bc - |\state_i^t - \state_j^ t|)$

\begin{equation}
\begin{split}
\frac{\partial \log \mathcal{L}}{\partial \bc} &= \frac{\partial}{\partial \bc} \left[ \sum_{(i, j, t) \in \order} \left( s_{ij}^t \log\kappa_{i, j, t}(\bc)+(1-s_{ij}^t)\log\left(1-\kappa_{i, j, t}(\bc)\right) \right) \right] \\
&= \sum_{(i, j, t) \in \order} s_{ij}^t\frac{\kappa_{i, j, t}(\bc) \left(1-\kappa_{i, j, t}(\bc)\right)}{\kappa_{i, j, t}(\bc)}-(1 - s_{ij}^t)\frac{\kappa_{i, j, t}(\bc) \left(1-\kappa_{i, j, t}(\bc)\right)}{1-\kappa_{i, j, t}(\bc)}\\
&= \sum_{(i, j, t) \in \order}s_{ij}^t - \kappa_{i, j, t}(\bc).
\end{split}
\end{equation}

By setting $\frac{\partial \log \mathcal{L}}{\partial \bc} = 0$, we obtain the relation
\begin{equation}
    \sum_{(i, j, t) \in \order}\kappa_{i, j, t}(\est\bc) = m,
\end{equation}
which 
corresponds to \Cref{eq:mle_bc}. 
This result demonstrates that the number of positive interactions is a sufficient statistic for the MLE of \bc.

The second derivative is given by
\begin{equation}
	\begin{split}
	\frac{\partial^2 \log \mathcal{L}}{\partial \bc^2} &= 
	-\frac{\partial \log \mathcal{L}}{\partial \bc}\sum_{(i, j, t)\in \order} \kappa_{i, j, t}(\bc)+m \\
	&= -\rho\sum_{(i, j, t) \in \order}\kappa_{i, j, t}(\bc) \left(1-\kappa_{i, j, t}(\bc)\right) %
	< 0 \quad \forall\ \bc ,
	\end{split}
\end{equation}
which proves that \Cref{eq:mle_bc} is a maximum of the likelihood function, implicitly defining the MLE for \bc.
\end{proof}

Intuitively, \Cref{eq:mle_bc} can be interpreted as considering all attempted interactions between pairs of agents, each one showing a distance $\lvert \state_i^t - \state_j^t \rvert$, and then finding $\est{\bc}$ by shifting all the corresponding sigmoids so that their sum---that is, the expected number of successful interactions---is equal to the observed one $m$. 
In practice, the parameter can easily be estimated by leveraging off-the-shelf optimization procedures applied to the likelihood functional.
As depicted by \Cref{fig:likelihood-example}, finding the maximum of this function is straightforward and allows us to find the true value of the confidence bound \tru{\bc}.

In our experiments, we use the secant method for finding the root when the derivative can be explicitly computed, and the Nelder-Mead minimization algorithm for the negative log-likelihood function when no derivative is provided, both functions are implemented in the SciPy Python package \cite{2020SciPy-NMeth}.

\begin{figure}
     \centering
         \includegraphics[width=0.8\textwidth]{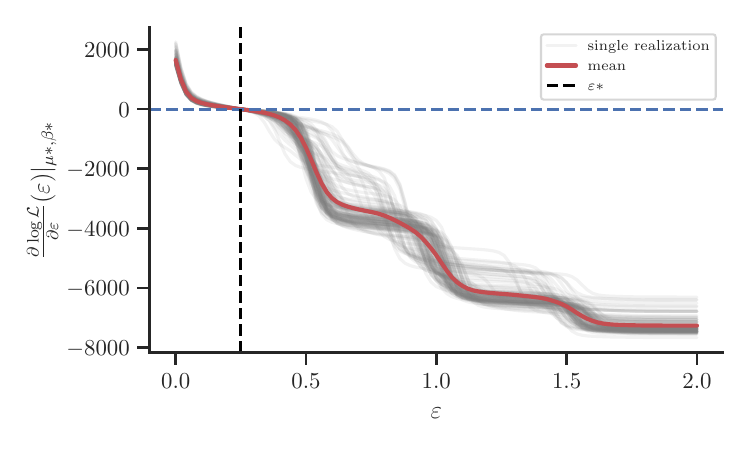}
        \caption{Example of $\partial\log\mathcal{L}(\Theta \mid \init, E)/\partial\bc$ over $\bc\in [0, 2]$. The derivative is monotonically decreasing. The red curve is the mean over all the realizations of the dynamics with fixed $\init, \order$.
        The derivative is zero for $\bc=\tru{\bc}$, i.e. the true value that originated the observed data.
        \label{fig:likelihood-example}
        }
\end{figure}

After deriving the MLE of \Cref{eq:mle_bc}, we study its bias $\mathbb{E}(\est{\bc}) - \bc$.
Our definition of sBCM allows us to exploit a formal equivalence with
Rasch models, a well-established framework in Item Response theory~\cite[Ch. 16]{fischer200616}.
In item response experiments, individuals respond to a set of items.
Rasch models characterize each individual with an ability parameter, and each item with a difficulty parameter.
Under the assumptions of Rasch models, the probability that an individual correctly answer to an item increases with the ability of the person, while it decreases with the item difficulty.

\begin{lemma}
    Estimating \bc in a sBCM with observed $\state^0$ and known \rate is equivalent to estimating the ability parameter of a single individual in a Rasch model.
    In particular, \bc is associated to the person parameter, and $|\state_i^t - \state_j^ t|$ is associated to the item parameter.
\label{lemma:rasch}
\end{lemma}

We provide the proof of \Cref{lemma:rasch} in \Cref{apx:rasch}, by simply substituting the variables of sBCM into the equation of Rasch models.
Now, we have all the ingredients to prove \Cref{theorem:asymptotically_unbiased}.
\begin{proof}[Proof of \Cref{theorem:asymptotically_unbiased}]
By leveraging well-known results on Rasch models, in \Cref{apx:biasvariance} we obtain the closed formula of the bias of the MLE.
We denote $h = (i,j,t)$ and $\kappa_h = \kappa_{i,j,t}(\bc)$.
\begin{equation}
    Bias(\est\bc) = \frac{1}{\rho \left(\sum\limits_{h \in \order} \kappa_{h} (1 - \kappa_{h}) \right)^2} \sum\limits_{h \in \order} \kappa_{h} (1 - \kappa_{h}) \left(\kappa_{h} - \frac{1}{2}\right)
    \label{eq:bias_rasch}
\end{equation}

Hence, the estimator is biased for small samples.
Specifically, the only term of \Cref{eq:bias_rasch} that can be negative is $\kappa_{h} - \frac{1}{2}$.
This implies that the bias is positive in the experiments with a higher weight given by positive interactions, consequently higher values of $\kappa_{h} - \frac{1}{2}$, and negative otherwise.
Therefore, the MLE pushes towards the extreme values,  overestimating high values of \bc, while underestimating low values of \bc.

Next, we can write
\begin{align}
\begin{split}
    \lvert Bias(\hat\bc) \rvert &= \left| \frac{1}{\rho \left(\sum\limits_{h \in \order} \kappa_{h} (1 - \kappa_{h}) \right)^2} \sum\limits_{h \in \order} \kappa_{h} (1 - \kappa_{h}) \left(\kappa_{h} - \frac{1}{2}\right) \right| \\
    &=  \frac{1}{\rho \left(\sum\limits_{h \in \order} \kappa_{h} (1 - \kappa_{h}) \right)^2} \sum\limits_{h \in \order} \kappa_{h} (1 - \kappa_{h}) \left(\right| \kappa_{h} - \frac{1}{2} \left| \right)  \\
    &<  \frac{1}{\rho \left(\sum\limits_{h \in \order} \kappa_{h} (1 - \kappa_{h}) \right)^2} \sum\limits_{h \in \order} \kappa_{h} (1 - \kappa_{h})  \frac{1}{2}  \\
    &=  \frac{1}{2 \rho \sum\limits_{h \in \order} \kappa_{h} (1 - \kappa_{h}) } \\
    &< \frac{1}{8 \rho T},    
\end{split}
\end{align}
and the bias is of order $T^{-1}$.
\end{proof}

Moreover, the bias goes to 0 as $\rho\rightarrow\infty$ (see \Cref{fig:rho_vs_bias}).
This means that as the sBCM tends to its deterministic counterpart, that is, Deffuant's BCM, the bias of MLE disappears.

Furthermore, we also have the estimate of the variance of the MLE (see \Cref{apx:biasvariance}), that is,

\begin{equation}
    Var(\est\bc) = \frac{1}{\rho^2 \sum\limits_{h \in \order} \kappa_{h} (1 - \kappa_{h})}.
\label{eq:var}
\end{equation}

\begin{figure}
     \centering
         \includegraphics[width=\textwidth]{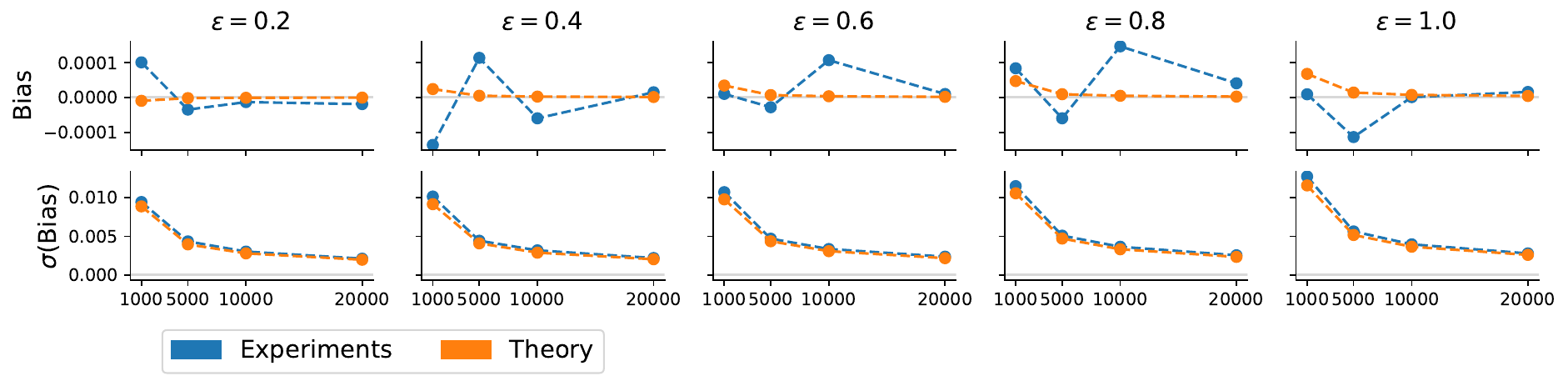}
        \caption{
        Comparison between theoretical bias and experimental bias.
        For each value in $\bc\in \{0.2, 0.4, 0.6, 0.8, 1.0\}$, for each $T \in \{1000, 5000, 10000,20000\}$, we run a set of \seedsrasch experiments, and we estimate \est\bc.
        In each experiment, we generate a trajectory of length $T$, with $1000$ agents, $\rate = 0.01$, and $\init$ uniformly distributed in $[-1,1]$.
        Note that the std. dev. of the bias is two orders of magnitude larger than the bias itself.
        }
        \label{fig:bias}
        
\end{figure}

In \Cref{fig:bias}, we compare the theoretical bias with the empirical bias obtained in a set of experiments.
In these experiments, we generate a data traces with sBCM, and estimate \est\bc.
First, we note that the bias decreases as $T$ grows, as well as the variance, and consequently the standard deviation.
From \Cref{fig:bias}, we notice that the standard deviation is order of magnitudes greater than the bias.
This result shows that the bias exists, but it is small compared to estimator variability.

\subsection{Estimation of \texorpdfstring{\rate}{Convergence Rate}}
\label{sec:anl-est-rate}

Let us shift our focus to estimating the rate parameter~\rate.
As illustrated by \Cref{fig:graphical-model}, we consider again $\state^0$ to be known, but subsequent states $\state^t$ with $t > 0$ are unknown: since \rate is to be estimated, in fact, future states cannot be computed simply from the sets of successful interactions.
In this section, we examine the likelihood concerning the trajectory of a single agent's opinion in the most basic scenario.
Through this approach, we analytically demonstrate some fundamental concerns in its formulation.
In particular, \Cref{lemma:lemma1,lemma:lemma2} show that both the support and the likelihood of the opinion under scrutiny depend on \rate.
This violates one of the regularity conditions underlying classical inferential MLE theory, which is typically leveraged to establish the consistency of MLE estimators.
This is in line with the results of the experiments in \Cref{sec:exp-est-rate}, which illustrate an upward bias in the MLE as it approaches its limit.



\spara{Biasedness.}
To analytically show the impossibility of employing classical statistical theory for the estimator of~\rate, 
let us consider a general setting in the simplest case, that is with the presence of two agents.
Thus, only one pair of agents can interact at each time step with a given probability.
We analyze the likelihood of the opinion trajectory of the agents in the first few time steps.  

More formally, consider two agents with opinions $\statevectorrv_i = (\staterv_i^0, \ldots,\staterv_i^T) , i=1, 2$ and denote their difference of opinion at time $t$ as 
\[
d_{12}^t = \staterv_2^t - \staterv_1^t.
\]
With a slight abuse of notation, we denote by $d_{12}^t$ also the realization of the distance between the two opinions at time step $t$, i.e.\ $d_{12}^t = \state_2^t -\state_1^t$.
Since we only consider two agents in our model, the set of (possible) interactions is fixed by construction: at each time step $t = 1,2, \ldots, T$, agent $1$ and $2$ can either interact (and change their opinions) or not interact.

Henceforth, we study the likelihood associated with the opinion of the first agent.
This choice is without loss of generality since the same calculations hold for the second agent just by swapping the indices.
\Cref{fig:opinion_trajectory} represents the evolution of the opinion of the first node, where the support at each time step is shown in the rectangles, while the (conditional) probabilities associated with each value are on the edges.
The random variable $\staterv_1^0$ possesses the Markov property, since its value at time $t+1$ conditional on its realization at time $t$ is independent of the previous realizations.
More interestingly, at each time step we can observe only one novel value for~$\staterv_1^t$: the opinion support $S^1_t$ for $t=1, 2, \ldots, T$ is such that $S^1_t \subset S^1_{t+1}$ and $\abs{S^1_{t+1}} = \abs{S^1_t} + 1$. 
Additionally, at each time step, only a single branch introduces a new opinion. Meanwhile, all the remaining branches maintain conditions already observed in the previous steps, thus essentially reiterating past dynamics in opinion shifts.
\begin{figure}
    \centering
    \scalebox{.7}{
    \begin{tikzpicture}[grow'=right,sibling distance=.25cm]

\tikzset{
level 1/.style={level distance = 5cm},
    level 2/.style={level distance = 7cm},
    level 3/.style={level distance = 8.5cm},
    every edge from parent/.style= 
            {edge from parent fork right},
         every tree node/.style=
            {draw,minimum width=3cm,text width=3cm,align=center},
            every node/.append style={midway}
            }
\Tree 
    [. $\staterv_1^0$ 
     \edge node[above left] {\tiny $\sigma_{\rho}(\bc - |d_{12}^0|)$}; 
        [. {$\staterv_1^0  + \rate d_{12}^0$ }
          \edge node[above] {\tiny $\sigma_{\rho}(\bc - |d_{12}^0|(1-2\rate))$}; 
            [.{$\staterv_1^0 + \rate (1-2\rate) d_{12}^0$ } 
            \edge node[above] {\tiny $\sigma_{\rho}(\bc - |d_{12}^0|(1-2\rate(1-2\rate))$}; 
            [. {$\staterv_1^0  + \rate d_{12}^0 [2-4\rate +4\rate^2)]$ } ]
            \edge node[below] {\tiny $1-\sigma_{\rho}(\bc - |d_{12}^0|(1-2\rate(1-2\rate))$}; 
            [. {$\staterv_1^0 + \rate (1-2\rate) d_{12}^0$ } ]
            ]
            \edge  node[below] {\tiny $1-\sigma_{\rho}(\bc - |d_{12}^0|(1-2\rate))$}; 
            [.{$\staterv_1^0  + \rate d_{12}^0$ }
             \edge node[above] {\tiny $\sigma_{\rho}(\bc - |d_{12}^0|(1-2\rate))$}; 
             [.{$\staterv_1^0 + \rate (1-2\rate) d_{12}^0$ } ]
            \edge  node[below] {\tiny $1-\sigma_{\rho}(\bc - |d_{12}^0|(1-2\rate))$}; 
            [.{$\staterv_1^0  + \rate d_{12}^0$ } ]
            ]
        ]
         \edge node[below] {\tiny $1-\sigma_{\rho}(\bc - |d_{12}^0|)$}; 
        [.$\staterv_1^0$ 
         \edge node[above] {\tiny $\sigma_{\rho}(\bc - |d_{12}^0|)$}; 
                [. {$\staterv_1^0  + \rate d_{12}^0$ } 
                \edge node[above] {\tiny $\sigma_{\rho}(\bc - |d_{12}^0|(1-2\rate))$}; 
            [.{$\staterv_1^0 + \rate (1-2\rate) d_{12}^0$ } ]
            \edge  node[below] {\tiny $1-\sigma_{\rho}(\bc - |d_{12}^0|(1-2\rate))$}; 
            [.{$\staterv_1^0  + \rate d_{12}^0$ } ]
                ]
                 \edge node[below] {\tiny $1-\sigma_{\rho}(\bc - |d_{12}^0|)$}; 
            [.{$\staterv_1^0$ }
             \edge node[above] {\tiny $\sigma_{\rho}(\bc - |d_{12}^0|)$}; 
                [. {$\staterv_1^0  + \rate d_{12}^0$ } ]
                  \edge node[below] {\tiny $1-\sigma_{\rho}(\bc - |d_{12}^0|)$}; 
                [. {$\staterv_1^0 $ } ]
            ]
        ] 
        ]
\end{tikzpicture}
}
    \caption{Opinion trajectory of the first agent, under the assumption of two nodes and fixed set of interactions.}
    \label{fig:opinion_trajectory}
\end{figure}
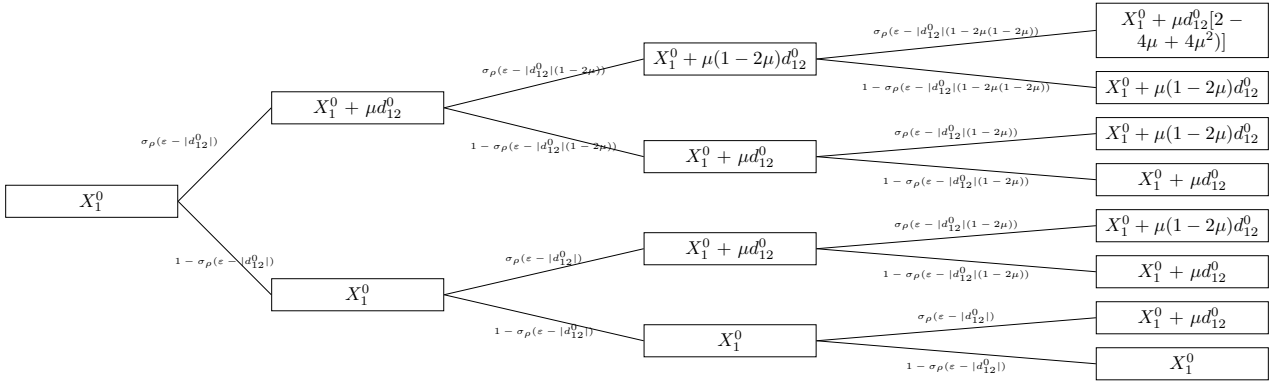
Given this framework, we can compute the likelihood associated with the opinion trajectory $\state_1^t,\ t=\{1, 2\}$.
\begin{lemma}\label{lemma:lemma1}
In the framework described above, the opinion of the first agent $\staterv_1^t$ at time $t = 1$, conditioned on~$\staterv_1^0 = \state_1^0$, is a Bernoulli random variable defined as
\begin{align}\label{eq:bernoulli1}
        \staterv_1^1 \mid (\staterv_1^0 = \state_1^0) = \begin{cases}
         \state_1^0 + \rate d_{12}^0 &\text{ with probability } \sigma_{\rho}(\bc - |d_{12}^0|), \\
         \state_1^0  &\text{ with probability } 1 - \sigma_{\rho}(\bc - |d_{12}^0|). 
        \end{cases}
\end{align}
Thus, the conditional likelihood becomes
\begin{equation}\label{eq:cond_lik1}
p( \staterv_1^1 \mid \staterv_1^0 = \state_1^0) =  \sigma_{\rho}(\bc - |d_{12}^0|)^{\mathbb{1}_{ \{ \staterv_1^1 = \state_1^0 + \rate d_{12}^0 \} }  } \left[1 - \sigma_{\rho}(\bc - |d_{12}^0|) \right]^{\mathbb{1}_{ \{ \staterv_1^1 = \state_1^0  \} }  }.
\end{equation}
\end{lemma}

\begin{lemma}\label{lemma:lemma2}
In the framework described above, the opinion of the first agent $\staterv_1^t$ at time $t = 2$, conditioned on~$\staterv_1^0 = \state_1^0$, is a Multinomial random variable defined as
\begin{align}\label{eq:multinomial2}
        \staterv_1^2  \mid (\staterv_1^0 = \state_1^0) = \begin{cases}
         \state_1^0 + \rate (1-2\rate) d_{12}^0 &\text{ w.p. } \sigma_{\rho}(\bc - |d_{12}^0|)\cdot \sigma_{\rho}(\bc - |d_{12}^0| (1-2\rate)), \\
         \state_1^0  + \rate d_{12}^0 &\text{ w.p. } \sigma_{\rho}(\bc - |d_{12}^0|)\cdot \left[1 - \sigma_{\rho}(\bc -  |d_{12}^0| (1-2\rate)) \right] + \\& +\left[1- \sigma_{\rho}(\bc - |d_{12}^0|) \right] \cdot \sigma_{\rho}(\bc - |d_{12}^0|), \\
         \state_1^0  &\text{ w.p. } [1- \sigma_{\rho}(\bc - |d_{12}^0|)]^2.
        \end{cases}
\end{align}
Thus, the conditional likelihood becomes
\begin{align}\label{eq:cond_lik2}
p( \staterv_1^2 \mid \staterv_1^0 = \state_1^0) = & [\sigma_{\rho}(\bc - |d_{12}^0|) \sigma_{\rho}(\bc - |d_{12}^0| (1-2\rate))]^{\mathbf{1}_{\{ \staterv_1^2 = \state_1^0 + \rate (1-2\rate) d_{12}^0 \} }}  \\
 \cdot &[\sigma_{\rho}(\bc - |d_{12}^0|) [1 - \sigma_{\rho}(\bc - |d_{12}^0| (1-2\rate)) ] + [1- \sigma_{\rho}(\bc - |d_{12}^0|)] \cdot \sigma_{\rho}(\bc - |d_{12}^0|)]^{\mathbf{1}_{\{ \staterv_1^2 = \state_1^0  + \rate d_{12}^0 \} }}  \\
 \cdot &  \{[1- \sigma_{\rho}(\bc - |d_{12}^0|)]^2\}^{\mathbf{1}_{\{ \staterv_1^2 = \state_1^0 \} }}.
\end{align}

\end{lemma}
\Cref{lemma:lemma1,lemma:lemma2} report the likelihoods of the opinion trajectory of one agent, given a simple dynamic involving only two agents.
Note that \Cref{lemma:lemma2} can be easily generalized for all the successive time steps~$t = 3, 4, \ldots, T$, just by expanding the support of the Multinomial opinion and recalculating the corresponding probabilities.
\Crefrange{eq:bernoulli1}{eq:cond_lik2} highlight a major issue that prevents us from applying standard maximum likelihood estimator theory~\cite{casella2002}: both the likelihood and the support of $\staterv_1^t$ depend on the parameter of interest~$\rate$, for $t \geq 1$.
This violates the standard regularity conditions under which classical theorems of MLEs hold, among which the consistency (see \citet[Section~10.6.2]{casella2002}).
In other words, we can not guarantee the consistency of the MLE for \rate using theorems from classical statistical inference.
This analysis holds for 2 agents; generalizing to $N$ agents is analytically complicated but doable in principle.
It does not, however, change the nature of the estimator: the MLE still provides no consistency guarantee.

Similar situations arise in some works in the literature~\citep{ferguson1982,hirano2003}.
These references demonstrate that the maximum likelihood estimator for a parameter lacks consistency and efficiency when the support of the modeled random variable depends on it.
Nevertheless, the methods proposed in these works are tailored and specific to particular circumstances, making their proofs and strategies non-generalizable to our context.

\section{Experiments}

In light of these analytical considerations, we perform experiments to assess the behavior of the MLEs for \bc and \rate in practice.
The experiments follow a simulation-based approach.
Our goal is to estimate empirically the error of the MLEs over the possible realizations of the model. 
To do so, we simulate the generative model by drawing uniformly at random $Q$ different values for the initial conditions \init.
For each of these values, we run the dynamics of the simulation as per \Cref{eq:pijt} for each time step.
Given that the model is probabilistic, we repeat the process $K$ times for each initial condition.
For each run, we apply the MLE to obtain an estimate of the studied parameters, \bc and \rate, and present summary statistics over the $K$ runs (to reduce stochastic noise).
By aggregating over the different $Q$ values, we estimate the expected distance (over the realizations of \init) of the MLE from the true value, as well as its variance.
For all the experiments we keep \order fixed, as, empirically, it does not have a significant impact.
As a robustness check, we use the same protocol by randomly drawing \order $Q$ times while keeping \init fixed. We obtain similar results across all the experiments.

We use the general setting described above in the following different configurations.
First, we apply this protocol to the estimate of \bc (\Cref{sec:exp-est-bc}) and then to the estimate of \rate (\Cref{sec:exp-est-rate}).
As stated in \Cref{sec:role-bc-rate}, these two cases follow different approaches:
the MLE of \bc is defined analytically, and to find it we solve $\sfrac{\partial\log\mathcal{L}}{\partial\bc} = 0$ numerically for \bc. Instead, to compute the MLE of \rate we directly maximize the likelihood numerically.
In both cases we vary the number of nodes of the graph $N$, and analyze the evolution of the estimates as a function of the length of the simulation $T$.
Then, \Cref{sec:exp-est-joint} analyzes the joint estimation of both parameters.
Other parameters are fixed
($\tru{\bc}=0.25$, $\tru{\rate}=0.5$, and $K=Q=100$).

\begin{figure}[ht!]
     \centering
     \begin{subfigure}[b]{0.49\textwidth}
         \centering
         \includegraphics[width=\textwidth]{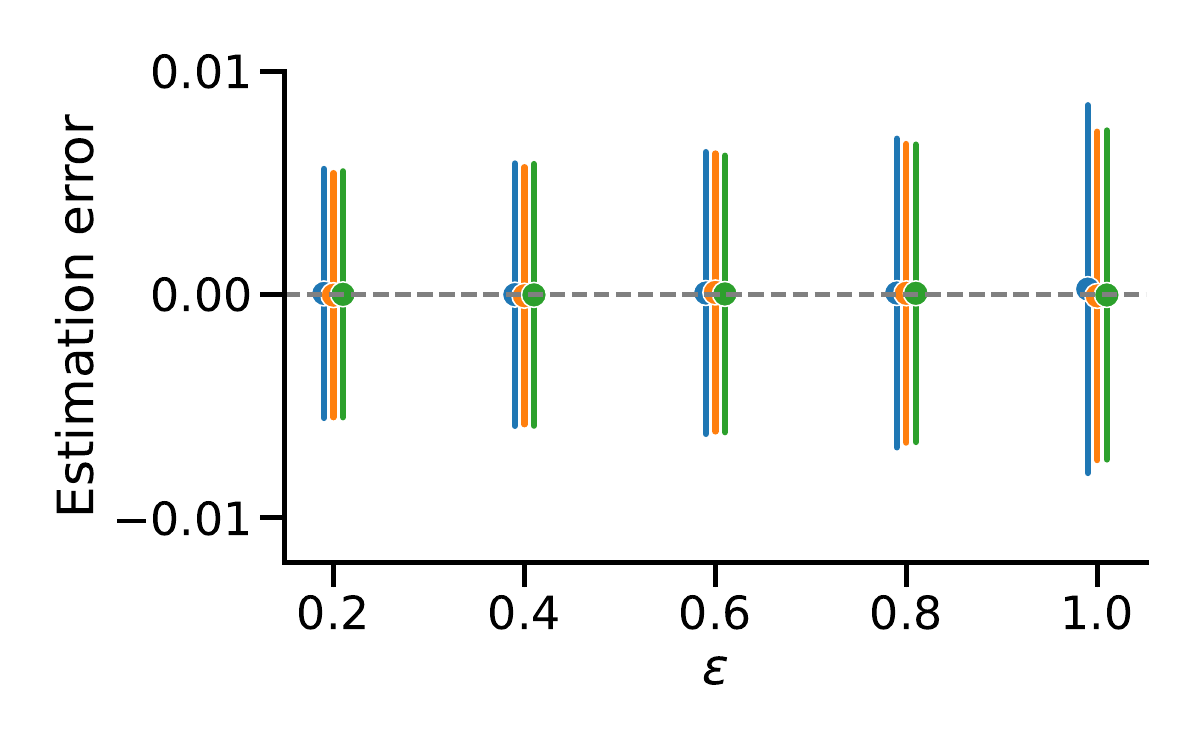}
         \caption{}
         \label{fig:minim_eps_mean}
     \end{subfigure}
     \hfill
     \begin{subfigure}[b]{0.49\textwidth}
         \centering
         \includegraphics[width=\textwidth]{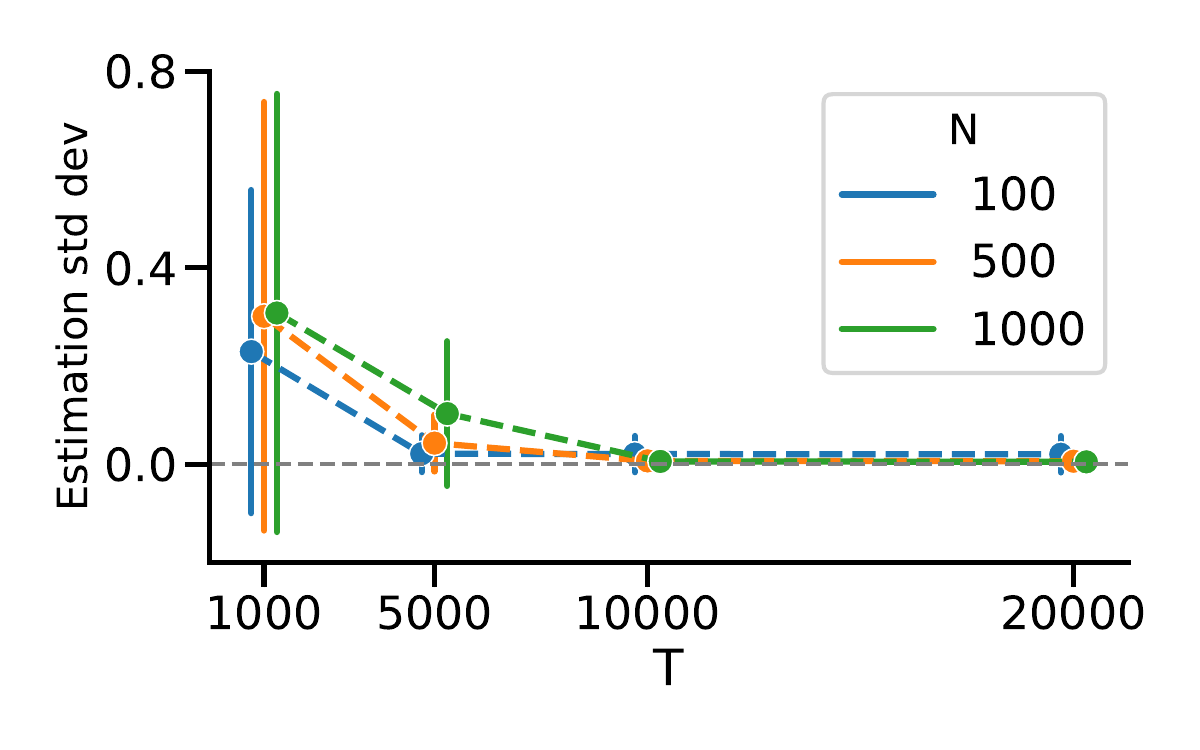}
         \caption{}
        \label{fig:minim_eps_std}
     \end{subfigure}
     \caption{(a) Estimation error for the bounded confidence parameter \est{\bc}. Error bars represent the std. dev. of the error. For high values of \bc the variance of the estimates increases and a small positive bias is visible for the small-scale scenario ($N=100$).
     (b) Std. dev. of the error for the bounded confidence parameter \est{\bc}. Error bars represent the standard error with
     \init fixed, $100$ different initialization of \order for $N = 100, 500, 1000$ and $T = 1000, \ldots, 20000$.  Repeated for $100$ values of \init.}
     
\end{figure}

\subsection{Estimation of \texorpdfstring{\bc}{Confidence bound}}
\label{sec:exp-est-bc}

The first goal of our experiments is to empirically verify that \est{\bc} is unbiased as found analytically in \Cref{sec:anl-est-bc}.
As such, we assume \tru{\rate} to be known.
\Cref{fig:minim_eps_mean} shows the expected estimation error over $Q$ different initial conditions.
Even for very small graphs~($N=10$) the error is of the order~$10^{-3}$, and quickly approaches zero as~$T$ grows.
For larger graphs the expected error of the estimate is almost zero even with a small number of samples.

Since we could not find an analytical form for the variance of this MLE, we resort to an empirical evaluation to assess it.
\Cref{fig:minim_eps_std} shows the variance of the estimator as a function of time and for different values of $N$.
Even for values of $N$ as small as $100$, the variance of the estimator is of the order of $10^{-2}$.
For larger networks, and when more samples are available (e.g., as $T$ increases), the estimator's variance decreases.
Overall, the variance of the MLE is small in absolute terms.

\subsection{Estimation of \texorpdfstring{\rate}{Convergence Rate}}
\label{sec:exp-est-rate}
In contrast to the estimation of \bc, \rate does not admit an explicit analytical form for the MLE. As already discussed in \Cref{sec:anl-est-rate}, the unbiasedness of the estimation cannot be guaranteed in the classical statistical theory framework.
In fact, we show experimentally that this lack of guarantees corresponds indeed to a bias: the estimation of \rate by numerical maximization of the likelihood generates a consistently biased estimation of the parameter. 
Similarly to the estimation of \bc, we consider in this case $\bc^*$ to be known and we obtain the MLE of~\rate by numerical optimization.
\Cref{fig:minim_mu_mean} shows the expected estimation error over $Q=100$ different initial conditions and realizations of the dynamics.
The figure shows a consistent, upward bias in the estimation.
This bias is approximately $10-15\%$ of the true value.
\Cref{fig:minim_mu_std} shows the standard deviation of the estimator, that is an order of magnitude larger than that of \bc even for large system sizes.
\begin{figure}[t]
     \centering
     \begin{subfigure}[b]{0.49\textwidth}
         \centering
         \includegraphics[width=\textwidth]{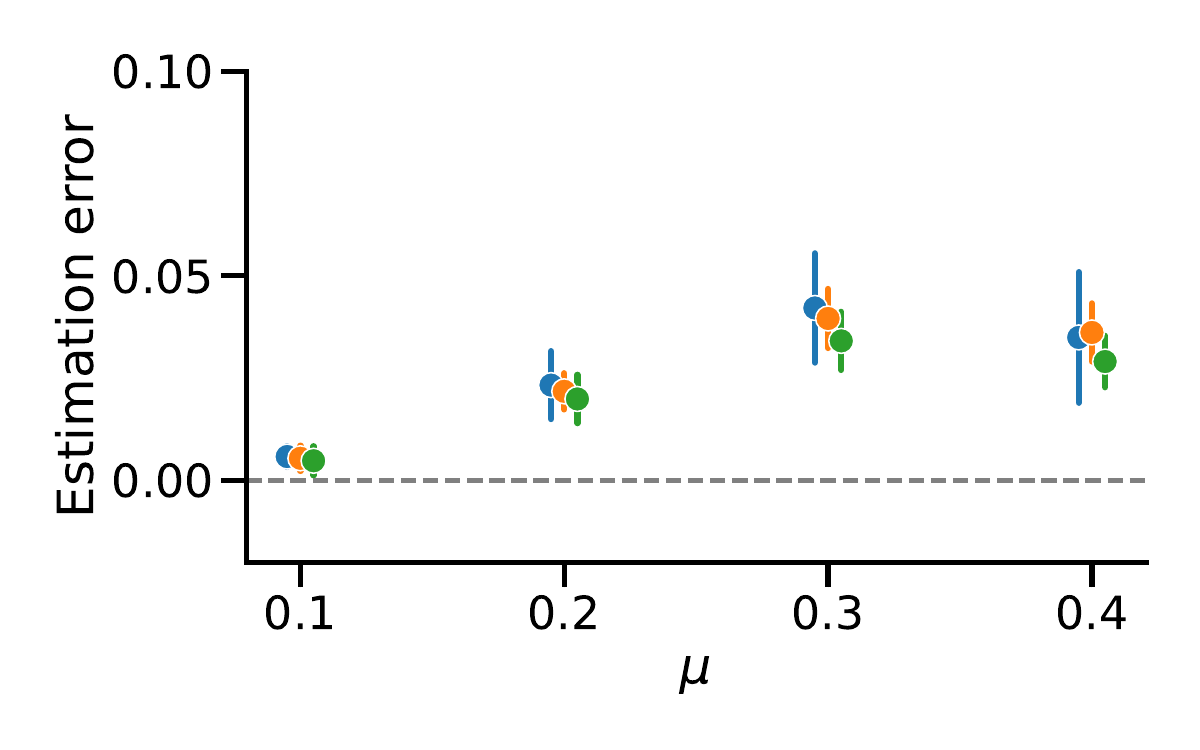}
         \caption{}
         \label{fig:minim_mu_mean}
     \end{subfigure}
     \hfill
     \begin{subfigure}[b]{0.49\textwidth}
         \centering
         \includegraphics[width=\textwidth]{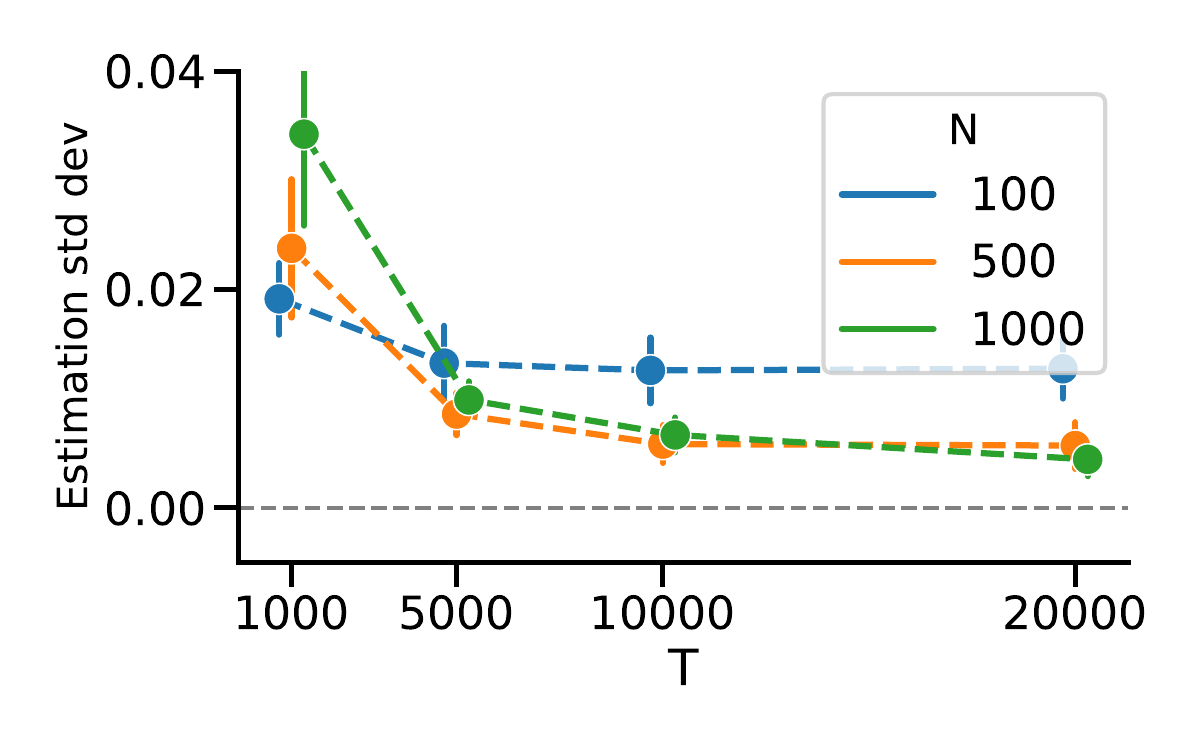}
         \caption{}
        \label{fig:minim_mu_std}
     \end{subfigure}
     \caption{(a) Estimation error for the bounded confidence parameter \est{\rate}. Error bars represent the std. dev. of the error. (b) Std. dev. of the error for the bounded confidence parameter \est{\bc}. Error bars
represent the standard error with \init fixed, $100$ different initializations of $T$ for $N = \{100, 500, 1000\}$ and
$T = \{1000, \ldots , 20000\}$. Repeated for $100$ values of \init.}

\end{figure}

Intuitively, this bias results from the fact that only sequences of successful interactions give any information about \rate, as unsuccessful ones do not change the state of the system.
Let us consider a pair of agents and their interactions over time as a sequence of zeros ($(i, j, t) \notin E$) or ones ($(i, j, t) \in E$).
Information about \rate comes only from the ones in the sequence:
for example, a sequence of zeros contains no information. 
To get information from the interactions at t, there must exist a one in $[0, t -1]$.
Consequently, we have more information from sequences with many ones.
At the same time though, the presence of ones increases the estimate of \rate, since successful interactions are more likely when opinions are closer.
Consequently, we tend to overestimate \rate, since we have more information for high \rate.

\subsection{Likelihood Shape and Identifiability}
\label{sec:exp-est-joint}

When both \bc and \rate are unknown, the same approach can be used for the joint estimation of the parameters. 
Similarly to the case of \rate, we do not have an explicit formula for computing the maximum of the likelihood function, and we thus rely on numerical minimization of the negative log-likelihood function. 

When considering the estimation of both parameters simultaneously, we observe a new interesting phenomenon.
As portrayed in \Cref{fig:joint_estimation_3d} (left panel), some regions of the parameter space behave as expected: the negative log-likelihood profile presents a unique minimum that can be easily computed via standard multivariate minimization methods.
However, there exist specific choices of the parameters that show a completely different likelihood profile (\Cref{fig:joint_estimation_3d} right panel).
In this case, even if the negative log-likelihood still presents a unique global minimum, other local minima can be present, and those minima are connected by a ``flat valley'' in the likelihood, thus introducing practical identifiability issues when little or noisy data is used to estimate the true parameters. 

\colorlet{darkred}{red!70!black}

\begin{figure}[t]
     \centering
     \begin{subfigure}[b]{0.49\textwidth}
         \centering
\includegraphics[width=0.8\textwidth,trim=0 0 150 0,clip]{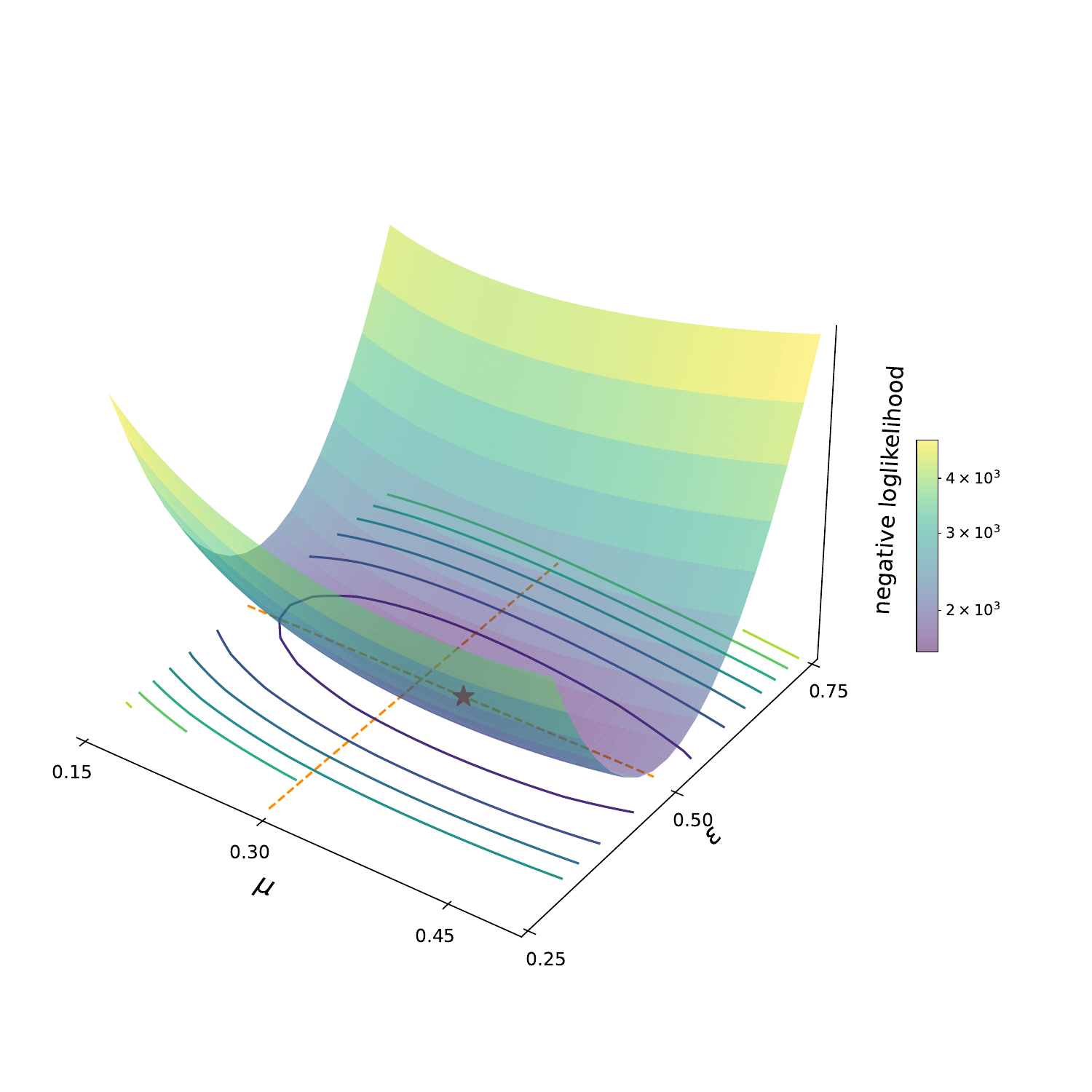}
         \caption{}
         \label{fig:minim_eps_mu_mean}
     \end{subfigure}
     \hfill
     \begin{subfigure}[b]{0.49\textwidth}
         \centering
         \includegraphics[width=\textwidth]{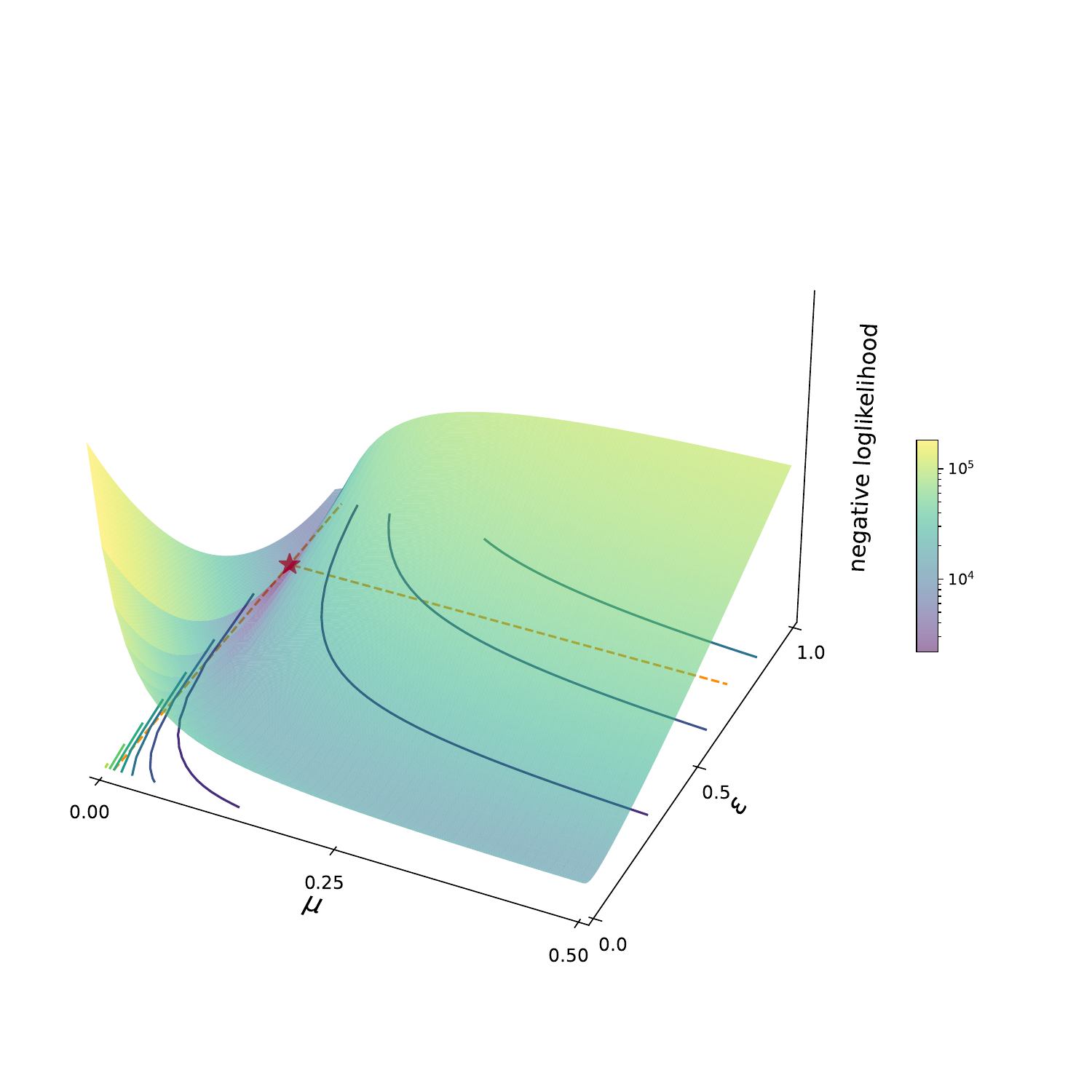}
         \caption{}
        \label{fig:minim_eps_mu_std}
     \end{subfigure}
     \caption{Negative Log Likelihood profile for the joint estimation of (\bc, \rate) for two different values of the true parameters. The true values are indicated by the dashed red lines. The value of the inferred minimum is indicated by a red star ($\textcolor{darkred}{\bigstar}$). The left panel shows a choice of parameters that generates a nll profile with a single but biased minimum, while the right panel shows a case where the system is characterized by possible practical identifiability issues.}
     \label{fig:joint_estimation_3d}
\end{figure}

\section{Discussion}
\label{sec:discussion}

In this work, we have formulated a fundamental question about the identifiability of the parameters of an agent-based model.
We have discussed the identifiability of parameters in the sBCM, and we have shown the vanishing bias in the estimator of the confidence bound~\bc, by connecting the model to theoretical results in Item Response Theory and Rasch models.
We have also shown a consistent bias in the convergence parameter~\rate due to the violation of fundamental assumptions for MLE.
Intuitively, successful interactions lead to convergence of opinions, which in turn increases the probability of success of an interaction:
this feedback loop induces a positive bias in the maximum-likelihood estimate of the parameter \rate.

In summary, even in a very simple model such as the BCM, with deterministic state transition functions, where all-but-one parameters and state variables of the model are observable at the most granular level, we incur in identifiability issues when estimating a single latent parameter.
These issues are likely to be compounded by increased complexity of ABMs and lower availability of data (e.g., using summary statistics rather than micro-level observations, or having latent state variables).

Considering the simplicity of the model we analyzed, our results should raise concerns about the implications for more complex agent-based models, such as those describing economic or biological systems.
Given the practical necessity of estimating these models with real-world data, prioritizing the study of identifiability to get meaningful inferences is paramount.
In fact, while it is common to use simulation-based approaches to estimate the parameters of complex ABMs~\cite{cranmer2020frontier}, there is a dearth of rigorous efforts to demonstrate the reliability of such estimation processes.
On the contrary, we argue that likelihood-based methods~\cite{lenti2024likelihoodbased} provide access to the conceptual toolbox of statistical inference to analyze such issues---as long as the ABM can be cast into a generative probabilistic model.
This work 
shows that in the ABM under scrutiny the violation of the standard MLE regularity conditions corresponds to a bias in the estimate.
Our findings sketch a general methodology to tackle the fundamental issue of identifiability and bias in the calibration of ABMs.
Further investigation of the statistical properties of MLE in the context of ABMs has the potential to yield efficient tools to provide guarantees on the calibration of model parameters, thereby spotting potential issues beforehand.

To achieve this goal, new formalisms to describe such agent-based models could prove helpful by providing a better-suited alternative to probabilistic graphical models.
In fact, in our case, the bias of the estimate of~\rate can not be identified directly from the graphical model of \Cref{fig:graphical-model}.
Not only \rate affects a state variable~\statevector in the graphical form, but it also affects the functional form of its support in \Cref{eq:modified_bcm} (which is not therein represented).
As an additional example, if \bc were to affect the functional form of the update rule, then also~\bc would not be ML-estimable.
More precisely, a \emph{sufficient} condition for having identifiability issues with a parameter $\phi$ in a sBCM-like model is that $\phi$ affects both state-space variables (i.e.\ their support) and the actual dynamic for the evolution of the opinions (i.e.\ the probabilities).
Indeed, in this case, at each time step both the support and the likelihood depend on such parameter.
This situation implies a violation of one of the regularity conditions of MLE (specifically, A3 in~\citeauthor{casella2002}~\cite{casella2002}), thus invalidating its usual guarantees.
Other assumptions of regularity could also be violated and hinder consistency in other ways.
For example, one could devise a model where multiple values of a parameter lead to the same outcome (observational equivalence).
Enhanced formal representations of ABMs should represent these aspects of the relationships among its variables.

In general, a more close-knit relation between the calibration of agent-based models and statistical inference is desirable~\cite{monti2023learning}.
This line of investigation could pave the way toward more ambitious goals, such as estimating micro-scale, individual, latent parameters and variables in ABMs, or performing model selection among different competing ABMs.
In the opinion dynamics model we considered, previous work~\cite{monti2020learning, lenti2024likelihoodbased} already identified the potential of likelihood-based methods for estimating its micro-parameters (the initial opinion of each agent in the system \init) from observable data.
Further work is needed to establish the guarantees offered by such estimates.
Proper initialization of agent-level parameters is essential in complex agent-based models, such as those used to simulate the epidemic and economic outcomes of COVID-19 under different scenarios~\cite{pangallo2024unequal}.
Proven and validated calibration of these parameters is essential for establishing agent-based models as reliable forecasting tools and unlocking new possibilities for predicting the evolution of real-world complex systems.

\section*{Acknowledgement}
 Valentina Ghidini was funded by the SNSF starting grant ``Multiresolution methods for unstructured data'' (TMSGI2\_211684).

\bibliographystyle{plainnat}
\bibliography{references}

\clearpage
\FloatBarrier

\section*{Supplementary Information}
\appendix

\section{Connection with Rasch models}
\label{apx:rasch}
Here we establish a connection between the sBCM framework and Rasch models (see \Cref{lemma:rasch}).
Rasch models are central in Item Response Theory, within psychometrics and educational assessments~\cite[Ch. 16]{fischer200616}.
In item response experiments, a set of individuals are asked to respond to a set of items.
Each individual $j$ is characterized by a personal ability, $\theta_j$, and each item is characterized by a difficulty parameter, $b_k$.
Under the assumptions of Rasch models, the probability that individual $j$ correctly answers item $k$ increases with the person's ability, while it decreases with the item difficulty.

\begin{proof}[Proof of \Cref{lemma:rasch}]
We denote the answer of individual $j$ to item $k$ as $u_{jk}$, and $u_{jk}$ is equal to 1 if the answer is correct, 0 otherwise.
In the most general form~\cite{lord1983unbiased}, the answer is modeled as a 3 parameter logistic model,

\begin{align}
    P(u_{jk} = 1 \mid A_k, b_k, c_k, \theta_k) = c_k + \frac{1 - c_k}{1 + \exp{(-A_k (\theta_j - b_k))}}.
\label{eq:rasch_prob_full}
\end{align}

In many cases, the model is simplified by assuming $c_k = 0$ and $A_k = 1$~\cite{fischer200616}, resulting in a 1 parameter logistic model,
\begin{align}
P(u_{jk} = 1 \mid \theta_j,  b_k ) = \sigma_{1}(\theta_j - b_k).
\label{eq:rasch_prob}
\end{align}

Afterwards, we focus on sBCM.
For consistency, we define $s_{ij}^t$ as the outcome of interaction the between $i$ and $j$ at time $t$.
If $(i,j,t)\in E$, $s_{ij}^t = 1$, otherwise $s_{ij}^t = 0$.
So, we rewrite \Cref{eq:pijt}, 
\begin{align}
P(s_{ij}^t = 1 \mid \bc, x_i, x_j ) = \sigma_{\rho}(\bc - \lvert x_i - x_j \rvert).
\label{eq:sbcm_prob}
\end{align}
We observe that \Cref{eq:rasch_prob_full} is transformed to \Cref{eq:sbcm_prob} by mapping calling $b_k \triangleq \lvert x_i - x_j \rvert$, $\theta_k \triangleq \bc$, $A_k \triangleq \rho$, and $c_k \triangleq 0$.
\end{proof}

\section{Bias and variance of \est\bc}
\label{apx:biasvariance}
\citeauthor{lord1983unbiased} prove the formula of the bias and variance of a Rasch model.
Based on the definition of \Cref{eq:rasch_prob_full} for a single person parameter ($\theta_j = \theta$), we define $P_i \triangleq P(u_{i} = 1 \mid A_i, b_i, c_i, \theta)$ and $Q_i \triangleq 1 - P_i$.
$P_i$ is the probability that the person with ability $\theta$ correctly responds to answer $i$, where $i \in \{1,\ldots,n\}$.
Since we are interested in sBCM, in the following we refer to the case where $c_i = 0$.

From the proof of \citeauthor{lord1983unbiased}, the variance of the MLE, which is equal to the variance of the bias, is 
\begin{equation}
    Var(\hat{\theta}) = \frac{1}{I} + o(n^{-1}),
\end{equation}

where $I$ is the Fischer information of $\theta$,
\begin{equation}
I = \sum\limits_{i=1}^n \frac{P_i'^2}{P_i Q_i},
    \label{eq:fischer}
\end{equation}

$P_i'$ is the derivative of $P_i$ with respect to $\theta$.

The bias is proven to be
\begin{equation}
Bias(\hat{\theta}) = \frac{1}{I^2} \sum\limits_{i = 1}^n A_i I_i \left(\phi_i - \frac{1}{2}\right),
    \label{eq:bias_lord}
\end{equation}
where 
\begin{equation}
    \phi_i = \frac{P_i - c_i}{1 - c_i},
\end{equation}
and 
\begin{equation}
I_i = \frac{P_i'^2}{P_i Q_i}.
    \label{eq:fischer_i}
\end{equation}

Therefore, we rewrite \Cref{eq:bias_lord} with the notation of sBCM.
Since $c_i = 0$, we have that $\phi_i = P_i = \kappa_h$.
We can also compute $P_i'$, resulting

\begin{equation}
P_i' = A_i P_i Q_i = \rho \kappa_h (1 - \kappa_h).
    \label{eq:derivative_p}
\end{equation}

By properly developing \Cref{eq:bias_lord}, we have
\begin{equation}
\begin{split}
    Bias(\hat{\theta}) &= \frac{1}{\left(\sum\limits_{h \in \mathcal{T}} \frac{(\rho \kappa_h (1 - \kappa_h))^2}{\kappa_h 1 - \kappa_h}\right)^2} \sum\limits_{h \in \mathcal{T}} \rho \frac{(\rho \kappa_h (1 - \kappa_h))^2}{\kappa_h (1 - \kappa_h)} \left(\kappa_h - \frac{1}{2}\right)\\
    &= \frac{1}{\rho \sum\limits_{h \in \mathcal{T}} \kappa_h (1 - \kappa_h)} \sum\limits_{h \in \mathcal{T}} \kappa_h (1 - \kappa_h) \left(\kappa_h - \frac{1}{2}\right).
\label{eq:derivation_bias_sbcm}
\end{split}
\end{equation}

This way, we obtain \Cref{eq:bias_rasch} of the main text.

\begin{figure}
    \centering
    \includegraphics[width=0.5\linewidth]{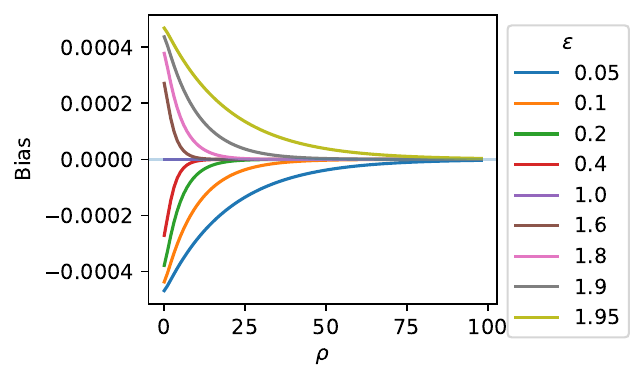}
    \caption{Theoretical bias in Rasch model, as a function of $\rho$. We compute the theoretical bias in experiments with 2000 items equally spaced in $[0,1]$.
    The plot shows that the bias goes to 0 as $\rho \rightarrow \infty$.
    As the model gets closer to the deterministic version of Rasch model, the bias vanishes.}
    \label{fig:rho_vs_bias}
\end{figure}

\end{document}